\documentclass[a4paper,11pt]{article}
\usepackage{amsthm}
\usepackage{cite}
\usepackage{indentfirst}
\usepackage{stmaryrd}

\usepackage[top=1 in,bottom=1 in,left=1.0 in,right=1.0 in]{geometry}

\newtheorem{lemma}{Lemma}[section]

\newtheorem{theorem}{Theorem}[section]
\newtheorem{proposition}{Proposition}[section]

\newtheorem{remark}{Remark}[section]

\def \b#1{\bar{#1}}
\def \deg{\mathrm{deg\,}}

\newcommand{\bsb}{\begin{subequations}}
\newcommand{\esb}{\end{subequations}}

\def \llbr{[\![}
\def \rrbr{]\!]}

\usepackage{amsmath}

\usepackage{amsfonts}
\usepackage{latexsym}
\usepackage{amssymb}
\usepackage{appendix}
\usepackage{mathrsfs}
\usepackage{graphicx}
\usepackage{bm}
\usepackage{cases}

\usepackage{color}
\newcommand{\bred}{\begin{color}{red}}
\newcommand{\ecl}{\end{color}}
\newcommand{\bblue}{\begin{color}{blue}}
\newcommand{\bgre}{\begin{color}{green}}
\newcommand{\bora}{\begin{color}{orange}}

\numberwithin{equation}{section}
			
\title{Symmetries of the D$\Delta$mKP hierarchy and their continuum limits}
\author{Jin Liu$^1$, ~ ~  Da-jun Zhang$^1$\footnote{Corresponding author. Email: djzhang@staff.shu.edu.cn},
~~ Xuehui Zhao$^{2,3}$\\
{\small  $^{1}$Department of Mathematics,
 Shanghai University, Shanghai 200444, China}\\
 {\small  $^{2}$College of Mathematical Science, Inner Mongolia Normal University,  Hohhot 010022, China}\\
 {\small  $^{3}$Center for Applied Mathematical Science, Inner Mongolia Normal University,
  Hohhot 010022, China}}

\date{\today}

\begin{document}
	
\maketitle

\begin{abstract}
In the recent paper [Stud. App. Math. 147 (2021) 752], squared eigenfunction symmetry constraint of
the differential-difference modified Kadomtsev-Petviashvili (D$\Delta$mKP) hierarchy
converts the D$\Delta$mKP system to the relativistic Toda spectral problem and its hierarchy.
In this paper we introduce a new formulation of independent variables
in the squared eigenfunction symmetry constraint,
under which the D$\Delta$mKP system gives rise to the discrete spectral problem and a hierarchy of
the differential-difference derivative nonlinear Schr\"odinger equation of the Chen-Lee-Liu type.
In addition, by introducing   nonisospectral flows,
two sets of  symmetries of the D$\Delta$mKP hierarchy and their algebraic structure are obtained.
We then present a unified continuum limit scheme,
by which we achieve  the correspondence of the mKP and the D$\Delta$mKP hierarchies
and their integrable structures.

\begin{description}
\item[Keywords:] D$\Delta$mKP hierarchy,  eigenfunction, symmetry,
constraint, derivative nonlinear Schr\"odinger, continuum limit
\item[PACS numbers:] 02.30.Ik, 02.30.Ks, 05.45.Yv
\end{description}
\end{abstract}

%

\section{Introduction}\label{sec-1}

The Kadomtsev-Petviashvili (KP) hierarchy and modified KP (mKP) hierarchy
are elementary continuous (2+1)-dimensional integrable systems.
They are respectively formulated as the compatible conditions of
\begin{subequations}\label{KP-lax}
\begin{align}
& \mathcal{L}\psi =\lambda \psi,~~
\mathcal{L}=\partial_x+u_2  \partial^{-1}_x +u_3 \partial^{-2}_x +\cdots, \label{KP-L}\\
& \psi_{t_m}=\mathcal{A}_m\psi,~~ \mathcal{A}_m= ( \mathcal{L}^m)_{\geq 0},
\end{align}
\end{subequations}
and
\begin{subequations}\label{mKP-lax}
\begin{align}
& L\phi =\lambda \phi,~~
L=\partial_x+v_0 +v_1 \partial^{-1}_x +v_2 \partial^{-2}_x +\cdots, \label{mKP-lax-L}\\
& \phi_{t_m}=A_m \phi, ~~ A_m=( L^m)_{\geq 1}.
\end{align}
\end{subequations}
Note that most of the involved notations can be referred from Sec.\ref{sec-2} of this paper.
It is well known that
the squared eigenfunction symmetry constraints bridge  a gap
between  (2+1)-dimensional and (1+1)-dimensional integrable systems,
which allows us to study (2+1)-dimensional equations using  (1+1)-dimensional ones,
and vice versa.
By  squared eigenfunction symmetry constraints,
\eqref{KP-lax} together with the adjoint form $\psi^*_{t_m}=-\mathcal{A}^*_m\psi^*$,
namely, the KP system, gives rise to the
Ablowitz-Kaup-Newell-Segur (AKNS) spectral problem and the AKNS hierarchy
\cite{Cao-1990,Cheng-1991,KSS-1991,KS-1991};
while for the mKP system composed by \eqref{mKP-lax} together with
the adjoint form $\phi^*_{t_m}=-A^*_m\phi^*$,
gives rise to the Kaup-Newell (KN) spectral problem and the derivative nonlinear Schr\"odinger (DNLS) hierarchy
of the Chen-Lee-Liu (CLL) type \cite{Chen-mKP-2002}.
More details about the above links can also be found in \cite{Chen-2006}.

For the KP system, such links have been extended to the differential-difference case,
namely, the differential-difference KP (D$\Delta$KP) system, composed by
\begin{subequations}\label{dKP-lax}
\begin{align}
& \bar{\mathcal{L}}\bar \psi_n =\lambda \bar\psi_n,~~
\bar{\mathcal{L}}=\Delta+\bar u_0 +  \bar u_1 \Delta^{-1} + \bar u_2 \Delta^{-2} +\cdots, \label{dKP-L}\\
& \bar\psi_{n,t_m}=\bar{\mathcal{A}}_m\bar\psi_n,~~
\bar \psi_{n,t_m}^*=\bar{\mathcal{A}}_m^*\bar\psi_n^*,~~
\bar{\mathcal{A}}_m= ( \bar{\mathcal{L}}^m)_{\geq 0}.
\end{align}
\end{subequations}
It turns out that, by the squared eigenfunction symmetry constraint,
the  D$\Delta$KP system is converted to the Merola-Ragnisco-Tu spectral problem and
a semi-discrete AKNS hierarchy \cite{Chen-JNMP-2017}.
Note that the Merola-Ragnisco-Tu spectral problem \cite{MRT-1994} (see also \cite{DJM-1983})
is a second discretization of the AKNS spectral problem (see \cite{Chen-JNMP-2017}),
which is also known as the
Darboux transformation of the AKNS spectral problem \cite{AY-JPA-1994},
but is different from the familiar discretization, namely, the Ablowitz-Ladik spectral problem \cite{AL-1975}.

In a recent paper \cite{CZZ-2021} the squared eigenfunction symmetry constraint for the
differential-difference modified KP (D$\Delta$mKP) hierarchy was investigated.
The later is related to the spectral problem
\begin{equation}
\bar{L}\bar \phi_n =\lambda \bar \phi_n,~~
\bar{L}=\bar v \Delta+\bar v_0 + \bar v_1 \Delta^{-1} + \bar v_2 \Delta^{-2} +\cdots. \label{dmKP-L1}
\end{equation}
It is surprising as well as interesting that with the constraint converts the D$\Delta$mKP system
to the spectral problem and hierarchy of  the  relativistic Toda,
not as we expected to the semi-discrete  DNLS system of the CLL type.

In the present paper we will give a new formulation of independent variables in
the squared eigenfunction symmetry constraint
such that the constrained  D$\Delta$mKP system can yield
a semi-discrete system that matches the DNLS system in a unified continuum limit.
Here by ``\emph{unified}'' we want to emphasize that
one should use a same continuum limit scheme to recover continuous objects from their
semi-discrete counterparts, including equation hierarchy, their integrable characteristics and their algebraic structures,
e.g.\cite{ZC-2010,FQSZ-2015} for the semi-discrete AKNS and \cite{FHTZ-2013,Chen-JNMP-2017}
for the D$\Delta$KP.
Note that sometimes  combinations of objects  are needed, e.g. \cite{FQSZ-2015,MP-1996}.
In this paper, apart from reinvestigating the  squared eigenfunction symmetry constraint,
we will also derive nonisospectral D$\Delta$mKP flows which not only provide a master symmetry
but also lead to a set of symmetries for the D$\Delta$mKP hierarchy.
Finally, we will provide a unified continuum limit scheme,
with which we will show the correspondence between all the above semi-discrete objects
and their continuous counterparts.

The paper is organized as follows.
Section \ref{sec-2} serves as the preliminary to introduce some notions and notations.
In Section \ref{sec-3}, we quickly review the  mKP hierarchy, their symmetries
and the squared eigenfunction symmetry constraint of the mKP system.
Then in Section \ref{sec-4}, we derive nonisospectral flows and symmetries for the D$\Delta$mKP heirarchy
and present their algebraic structures.
In addition, as the main result  of the paper, we give a new formulation of the squared eigenfunction
symmetry constraint and derive the differential-difference CLL spectral problem and
the differential-difference DNLS hierarchy from the constrained D$\Delta$mKP system.
After that, in Section \ref{sec-5}, we present a unified scheme of the continuum limits and show that
the obtained results and their algebraic structures can
well match their continuous  counterparts given in Section \ref{sec-3}.
Finally, in Section \ref{sec-6}, conclusions are given.
There are two appendices. One provides a detailed proof of formula \eqref{vs-rhs},
and in the another, we discuss discrete CLL spectral problem and the related discrete integrable systems.

\section{Preliminary}\label{sec-2}

We briefly list some notions and notations used in our paper.
One can also refer to  \cite{FF-1981} and \cite{FHTZ-2013}.
Let $v= v(x,y,t)$ be a $C^{\infty}$ function defined on $\mathbb{R}^3$  and decrease rapidly
as $|x|, |y| \to \infty$.
By $S[v]$ we denote a Schwartz space composed by all functions $f(v)$
that are $C^{\infty}$ differentiable with respect to $v$ and its derivatives.
For two functions $f,g \in S[v]$, the G$\hat{\mathrm{a}}$teaux derivative of $f$ with respect to $v$ in direction $g$
is defined as
\begin{equation*}
f'[g]\doteq f'(v)[g]=\frac{\mathrm{d}\, f(v+\varepsilon g)}{\mathrm{d}\,\varepsilon}\Big|_{\varepsilon=0},
\end{equation*}
using which a Lie product $\llbracket \cdot,\cdot \rrbracket$ is defined as
\begin{equation}\label{lie}
\llbracket f,g\rrbracket=f'[g]-g'[f].
\end{equation}
For an  evolution equation
\begin{equation}\label{vt-K}
v_t=K(v)
\end{equation}
where $K(v)\in S[v]$,  we say $\omega=\omega(v)$ is its symmetry
if for all solutions $v$ of equation \eqref{vt-K} there is
\begin{equation*}
\omega_t=K'[\omega].
\end{equation*}
In the pseudo-differential operator $\mathcal{L}$ given in \eqref{KP-L},
$\partial_x=\frac{\partial}{\partial x}$ and $\partial^{-1}_x \partial_x= \partial_x \partial^{-1}_x=1$.
For two functions $f,g \in S[v]$, their inner product is defined as
\begin{equation*}
\langle f, g \rangle=\int_{-\infty}^{\infty}\int_{-\infty}^{\infty} fg \, \mathrm{d}x \mathrm{d}y.
\end{equation*}
For an operator $T$ living on $S[v]$, its adjoint operator $T^*$ is defined through
\begin{equation}\label{inn}
\langle f,T g\rangle=\langle T^* f,g\rangle.
\end{equation}

In the semi-discrete (differential-difference) case,
suppose that $\tilde{v}=\tilde{v}(n,\bar x,\bar t)$ is a function
of $(n,\bar x,\bar t) \in \mathbb{Z} \times \mathbb{R}^2$, $C^{\infty}$ with respect to $(\bar x, \bar t)$,
and decreases rapidly as $|n|, |\bar x| \to \infty$.
Similar to the continuous case, one can introduce  a Schwartz function space $\bar S[\tilde v]$
and define G$\hat{\mathrm{a}}$teaux derivative, Lie product, and symmetry
for an evolution equation $\tilde v_{\bar t}=K(\tilde v)$.
Without confusion we write $f(\tilde{v})=f_n$.
By $\Delta$ and $E$ we denote the difference operator and shift operator with respect to $n$,
i.e. $\Delta f_n=(E-1) f_n=f_{n+1}-f_n$. There is an extended Leibniz rule for $\Delta$ \cite{ES-2011},
\begin{equation}\label{2.7}
\Delta^s g_n=\sum_{i=0}^{\infty} \mathrm{C}_s^i(\Delta^i g_{n+s-i})\Delta^{s-i},~~~ s\in \mathbb{Z},
\end{equation}
where
\begin{equation*}
\mathrm{C}_0^0=1,~~ \mathrm{C}_s^i=\frac{s(s-1)(s-2)\cdots(s-i+1)}{i!}.
\end{equation*}
An inner product of two functions in $\bar S[\tilde v]$ is defined as
\begin{equation}\label{2.8}
\langle f_n,g_n\rangle=\sum_{n=-\infty}^{\infty}\int_{-\infty}^{\infty}f_n g_n\, \mathrm{d}\bar x,
\end{equation}
and for an operator $T$  its adjoint operator $T^*$ is defined as
$\langle f_n, \bar T g_n\rangle=\langle \bar T^* f_n,g_n\rangle$.

\section{Results of the mKP}\label{sec-3}

In the following we quickly review the  mKP hierarchy, their symmetries and algebraic structures
(see also \cite{DCZ-2017}),
and the squared eigenfunction symmetry constraint of the mKP system (see also \cite{Chen-mKP-2002,Chen-2006}).
These will be finally recovered from the results of the D$\Delta$mKP using a unified continuum limit.

\subsection{The scalar mKP hierarchy and symmetries}\label{sec-3-1}

Consider the Lax triad\footnote{We should consider $y$ and $t_2$ independently,
especially in the nonisospectral case, see \cite{FHTZ-2013}.} \cite{DCZ-2017}
\begin{subequations}\label{mkp-lax-iso}
\begin{align}
&L\phi=\lambda \phi, ~~~\lambda_{t_m}=0,\\
&\phi_y=A_2 \phi,\\
&\phi_{t_m}=A_m \phi, ~~ (m=1,2,\cdots),
\end{align}
\end{subequations}
where $L$ is the pseudo-differential operator given in \eqref{mKP-lax-L}
with $v_0\equiv v,~ v_j\in S[v]$, $A_m=(L^m)_{\geq 1}$ which contains the pure differential part of $L^m$,
e.g.,
\begin{subequations}\label{mkp-am}
\begin{align}
&A_1=\partial_x,\\
&A_2=\partial^2_x+2v_0\partial_x,\\
&A_3=\partial^3_x+3v_0\partial^2_x+3(v_1+v_{0,x}+v_0^2)\partial_x.
\end{align}
\end{subequations}
Compatibility of the triplet \eqref{mkp-lax-iso} gives rise to
\begin{subequations}
	\begin{align}
		&L_y=[A_2,L],\label{mkp-iso-Ly}\\
		&L_{t_m}=[A_m,L],\label{mkp-iso-Ltm}\\
		&A_{2,t_m}-A_{m,y}+[A_2,A_m]=0,\label{mkp-iso-zcc}
	\end{align}
\end{subequations}
where $[A,B]=AB-BA$.
Among them Eq.\eqref{mkp-iso-Ly} plays a role to express $v_j$ in terms of $v_0\equiv v$ for $j=1,2,\cdots$,
Eq.\eqref{mkp-iso-Ltm} agrees with $A_m=(L^m)_{\geq 1}$ with the boundary condition
$A_m |_{\boldsymbol{v=0}}=\partial^m_x$ where $\boldsymbol{v}=(v_0, v_1, v_2, \cdots)$
(see the following Remark \ref{Rem-3-1}),
and the third equation \eqref{mkp-iso-zcc} gives rises to the zero curvature representation
of the scalar isospectral mKP hierarchy
\begin{equation}\label{mkp-h}
v_{t_m}=K_m=\frac{1}{2}(A_{m,y}+[A_m,A_2])\partial^{-1}_x,~~~m=1, 2, \cdots.
\end{equation}
Here we list out the first three equations in the hierarchy:
\begin{subequations}\label{3.5}
\begin{align}
&v_{t_1}=K_1=v_x,\\
&v_{t_1}=K_2=v_y,\\
&v_{t_3}=K_3=\frac{1}{4}v_{xxx}-\frac{3}{2}v^2 v_x+\frac{3}{2}v_x \partial^{-1}_x v_y +\frac{3}{4}\partial^{-1}v_{yy},
\end{align}
\end{subequations}
where the third equation is known as the mKP equation.

\begin{remark}\label{Rem-3-1}
If we assume $A_m$ has a form
\[A_m=  a_0  \partial_x^m+  a_1  \partial_x^{m-1}+ \cdots + a_{m-1}  \partial_x,~~
A_m |_{\boldsymbol{v=0}}=\partial^m_x,
\]
then $\{a_j\}$ can be uniquely determined from \eqref{mkp-iso-Ltm} and it turns out that $A_m=(L^m)_{\geq 1}$.
See also \cite{ZC-2003}.
\end{remark}

The nonisospectral mKP flows are formulated from the triad \cite{DCZ-2017}
\begin{subequations}\label{mkp-lax-non}
\begin{align}
&L\phi=\lambda \phi, ~~~\lambda_{t_m}=\lambda^{m-1},\\
&\phi_y=A_2 \phi,\\
&\phi_{t_m}=B_m \phi, ~~ (m=1,2,\cdots),
\end{align}
\end{subequations}
where $L$ and $A_2$ are the same as the isospectral case, while $B_m$ is a differential operator
\begin{equation*}
B_m=b_0\partial^m_x+b_1\partial^{m-1}_x+\cdots+b_{m-1}\partial_x,
\end{equation*}
with assumption
\begin{equation}\label{Bm-b}
B_m|_{\boldsymbol{v}=0}=
\begin{cases}
2y\partial_x, & m=1,\\
2y\partial^m_x+x\partial^{m-1}_x,& m \ge 2.
\end{cases}
\end{equation}
The compatibility of \eqref{mkp-lax-non} leads to
\begin{subequations}
\begin{align}
&L_y=[A_2,L],\label{mkp-non-Ly}\\
&L_{t_m}=[B_m,L]+L^{m-1},\label{mkp-non-Ltm}\\
&A_{2,t_m}-B_{m,y}+[A_2,B_m]=0,\label{mkp-non-zcc}
\end{align}
\end{subequations}
among which, \eqref{mkp-non-Ly} is the same as \eqref{mkp-iso-Ly}
which means $v_j$ are  expressed in terms of $v_0\equiv v$ as same as the isospectral case;
by \eqref{mkp-non-Ltm} one can uniquely determine $B_m$ with the boundary condition \eqref{Bm-b},
and we have
\begin{align*}
&B_1=2y\partial_x,\\
&B_2=2yA_2+xA_1,\\
&B_3=2yA_3+xA_2,
\end{align*}
etc, where  $A_j$ are given in \eqref{mkp-am};
Eq.\eqref{mkp-non-zcc} provides the zero curvature representation of the scalar nonisospectral mKP hierarchy
\begin{equation}
u_{t_m}=\sigma_m=\frac{1}{2}(B_{m,y}+[B_m,A_2])\partial^{-1}_x,
\end{equation}
where the first few equations are
\begin{subequations}
\begin{align}
&v_{t_1}=\sigma_1=2yK_1+1,\\
&v_{t_2}=\sigma_2=2yK_2+xK_1+v,\\
&v_{t_3}=\sigma_3=2yK_3+xK_2-\frac{1}{2}v^2+\frac{3}{2}\partial^{-1}_x v_y-\frac{1}{2}v_x,
\end{align}
\end{subequations}
with $K_j$ given in \eqref{3.5}.
Note that asymptotically we have
\begin{equation}
K_m|_{v=0}=0,
~~~
\sigma_m|_{v=0}=
\begin{cases}
1, & m=1,\\
0,& m \ge 2.
\end{cases}
\end{equation}

The flows $\{K_m, \sigma_s\}$ compose a Lie algebra with respect to the  product
$\llbracket \cdot, \cdot \rrbracket$ given in \eqref{lie} \cite{DCZ-2017}.

\begin{theorem} \label{Th-1}
The mKP flows $\{ K_m, \sigma_s \}_{m\geq 1, n \geq 1}$ compose a  Lie algebra with the following structure
\begin{subequations}  \label{flow-alg-e}
\begin{eqnarray}
    &&  \llbr K_m, K_s \rrbr  = 0,           \label{flow-alg-e kk}       \\
    &&  \llbr K_m, \sigma_s \rrbr  =    m K_{m+s-2},         \label{flow-alg-e ksigma}   \\
    &&  \llbr \sigma_m, \sigma_s \rrbr   =    (m-s) \sigma_{m+s-2},      \label{flow-alg-e ssigma}
\end{eqnarray}
\end{subequations}
in which we specially define $K_0=\sigma_0=0$.
Note that $\sigma_3$ plays a role of master symmetry.
\end{theorem}

The above structure immediately yields symmetries of the mKP hierarchy \cite{DCZ-2017}.
\begin{theorem}\label{Th-2}
Each equation
\begin{equation}
v_{t_m}=K_m
\end{equation}
in the isospectral mKP hierarchy \eqref{mkp-h} has two sets of symmetries:
\begin{equation}  \label{corolary2  symmetries}
  \{ K_s \}_{s \geq 1}, \qquad \{\tau_k^m= m\, t_m K_{m+k-2}+\sigma_k \}_{k \geq 1},
\end{equation}
and they compose a Lie algebra with structure
\begin{subequations}
\begin{align}
    &  \llbr K_s, K_k \rrbr  =      0,                \\
    &  \llbr K_s, \tau_k^m \rrbr  =  s\, K_{s+k-2},                        \\
    &  \llbr \tau_s^m, \tau_k^m \rrbr   =  (s-k) \tau^m_{s+k-2},
\end{align}
\label{flow alge-1}
\end{subequations}
where we define $K_0=\tau^m_0=0$.
\end{theorem}

\subsection{Squared eigenfunction symmetry constraint}\label{sec-3-2}

In addition to the Lax triad \eqref{mkp-lax-iso}, we consider the following adjoint system
\begin{subequations}\label{mkp-lax-iso*}
\begin{align}
&L^*\phi^*=\lambda \phi^*, ~~~\lambda_{t_m}=0,\\
&\phi_y^*=-A_2^* \phi^*,\\
&\phi_{t_m}^*=-A_m^* \phi^*, ~~ (m=1,2,\cdots),
\end{align}
\end{subequations}
where $\phi^*$ stands for the eigenfunction,
$L^*$ and $A_m^*$ are the adjoint operators of $L$ and  $A_m$ defined through \eqref{inn}.
One can prove that the above system yields the iosospectral mKP hierarchy $v_{t_m}=K_m$ as well.
In addition, it can be proved that
\begin{equation}
\omega=(\phi\phi^*)_x
\end{equation}
is a (squared eigenfunction) symmetry of the iosospectral mKP hierarchy \eqref{mkp-h},
provided $\phi$ and $\phi^*$ are solutions to the
triplets \eqref{mkp-lax-iso} and \eqref{mkp-lax-iso*}, respectively.

Since $v_x$ is also a symmetry of the iosospectral mKP hierarchy, by imposing
$v_x=(\phi\phi^*)_x$ we have a squared eigenfunction symmetry constraint \cite{Chen-mKP-2002}
\begin{equation}
v=\phi\phi^*=qr,~~ (\phi=q,~ \phi^*=r).
\end{equation}
With such a constraint,
the pseudo-differential operator $L$ can be consistently written as
\begin{equation}
L=\partial_x + q \partial_x^{-1} r \partial_x,
\end{equation}
and consequently the spectral problem $L\varphi = \lambda \varphi$ gives rise to
\begin{equation}\label{KN-sp1}
\left(\begin{array}{c}
\varphi_1\\
\varphi_2
\end{array}
\right)_x
=\left(\begin{array}{cc}
-\eta^2 & \eta q\\
\eta r & -qr
\end{array}
\right)
\left(\begin{array}{c}
\varphi_1\\
\varphi_2
\end{array}
\right),
\end{equation}
where we have taken $\varphi_1=\varphi$ and $\lambda=-\eta^2$.
Introducing a gauge transformation
\begin{equation*}
\varphi_i=e^{-\frac{1}{2}(\eta^2 x+ \partial^{-1}_x qr)} \psi_i,~~ (i=1,2),
\end{equation*}
the spectral problem \eqref{KN-sp1} yields a SL(2) version
\begin{equation}\label{KN-sp2}
\left(\begin{array}{c}
\psi_1\\
\psi_2
\end{array}
\right)_x
=\left(\begin{array}{cc}
-\frac{1}{2}(\eta^2-qr) & \eta q\\
\eta r & \frac{1}{2}(\eta^2-qr)
\end{array}
\right)
\left(\begin{array}{c}
\psi_1\\
\psi_2
\end{array}
\right),
\end{equation}
which is the KN spectral problem to generate the DNLS hierarchy of the Chen-Lee-Liu type \cite{WS-JPSJ-1983}.
In the following we call it the CLL spectral problem for short.
In addition, the coupled system $\phi_{t_m}= A_m  \phi,~  \phi_{t_m}^*=-A_m^* \phi^*$
gives rise to  the DNLS hierarchy of the Chen-Lee-Liu type
(i.e. the CLL hierarchy)
\begin{equation}\label{CLL-hie}
\left(\begin{array}{c}
q\\
r
\end{array}
\right)_{t_{m+1}}
=T \left(\begin{array}{c}
q\\
r
\end{array}
\right)_{t_{m}},~~~
\left(\begin{array}{c}
q\\
r
\end{array}
\right)_{t_{1}}=
\left(\begin{array}{c}
q\\
r
\end{array}
\right)_{x},~~ (m=1,2,\cdots),
\end{equation}
where the recursion operator $T$ is
\begin{equation}\label{T}
T=\left(\begin{array}{cc}
\partial_x+qr+q_x\partial_x^{-1}r-q\partial_x^{-1}r_x & q_x\partial_x^{-1}q+q\partial_x^{-1}q_x\\
r_x\partial_x^{-1}r+r\partial_x^{-1}r_x & -\partial_x+qr+r_x\partial_x^{-1}q-r\partial_x^{-1}q_x
\end{array}
\right).
\end{equation}
For more details of the above squared eigenfunction symmetry constraint,
one can refer to \cite{Chen-mKP-2002,Chen-2006}.

\section{The  D$\Delta$mKP: symmetries and constraint}\label{sec-4}

Since in the recent paper \cite{CZZ-2021} the squared eigenfunction symmetry constraint leads the
D$\Delta$mKP system to the spectral problem and the positive hierarchy of the relativistic Toda,
not those of the DNLS as we expected,
our particular interest in the current paper is to find a new formulation of the squared eigenfunction symmetry constraint
so that the constrained the D$\Delta$mKP system matches the
continuous results described in Sec.\ref{sec-3}.
In addition, the match of the algebraic structures in the continuum limit
is also important because this means the discretization keeps well the integrable structures.
In the following we will first derive nonisospectral  D$\Delta$mKP  flows,
from which we will obtain symmetries of the D$\Delta$mKP hierarchy and their Lie algebraic structure.
After that, we will look at the squared eigenfunction symmetry constraint.

\subsection{Symmetries and algebraic structures}\label{sec-4-1}

In the following, to consider continuum limit, we introduce the spacing parameter $h$ in the $n$-direction.
Thus, instead of the pseudo-difference operator $\bar{L}$ given in \eqref{dmKP-L1}, we use the following form,
\begin{equation}\label{dmkp-L}
\bar{L}=h^{-1}\bar{v}\Delta+\bar{v}_0+h\bar{v}_1\Delta^{-1}+\cdots+h^j \bar{v}_j \Delta^{-j}+\cdots,
\end{equation}
where
\begin{equation}
\bar{v}=1+h \tilde v
\end{equation}
and we assume $\tilde v$ and $\bar{v}_i$ belong to $\bar S[\tilde{v}]$.

Consider the isospectral Lax triad (cf.\cite{CZZ-2021})
\begin{subequations}\label{dmkp-iso-lax}
\begin{align}
&\bar{L} \bar \phi=\lambda \bar \phi,~~~\lambda_{t_m}=0, \label{dmkp-sp}\\
&\bar \phi_{\bar{x}}=\bar{A}_1\bar \phi,~~~\bar{A}_1=h^{-1}\bar{v}\Delta, \label{dmkp-iso-A1}\\
&\bar \phi_{\bar t_m}=\bar{A}_m\bar \phi,~~~m=1, 2, \cdots, \label{dmkp-flow-lax}
\end{align}
\end{subequations}
where $\bar{A}_m=(\bar{L}^m)_{\geq 1}$ contains the pure difference part (in terms of $\Delta$)
of $\bar{L}^m$, and the first three are
\begin{subequations}\label{Aj}
\begin{align}
&\bar{A}_1=h^{-1}\bar{v}\Delta,\label{A1}\\
&\bar{A}_2=h^{-2}\bar{v}(E\bar{v})\Delta^2+h^{-2}\bar{v}(\Delta \bar{v})\Delta+h^{-1}\bar{v}(E\bar{v}_0)\Delta+h^{-1}\bar{v}\bar{v}_0\Delta,\\
&\bar{A}_3=h^{-3} \bar{a}_{0} \Delta^3+h^{-2} \bar{a}_1 
\Delta^2+h^{-1} \bar{a}_2 \Delta,
\end{align}
\end{subequations}
in which
\begin{align*}
 \bar{a}_0  &=\bar{v}(E\bar{v})(E^2\bar{v}),\\
 \bar{a}_1 &=2h^{-1}\bar{v}(E\bar{v})(E\Delta\bar{v})+\bar{v}(E\bar{v})(E^2\bar{v}_0)
 +h^{-1}\bar{v}(\Delta \bar{v})(E\bar{v})+\bar{v}(E\bar{v}_0)(E\bar{v})+\bar{v}\bar{v}_0(E\bar{v}),\\
 \bar{a}_2 & = h^{-2}\bar{v}(E\bar{v})(\Delta^2\bar{v})
+2h^{-1}\bar{v}(E\bar{v})(E\Delta\bar{v}_0)+\bar{v}(E\bar{v})(E^2\bar{v}_1)
+h^{-2}\bar{v}(\Delta\bar{v})^2+2h^{-1}\bar{v}(\Delta\bar{v})(E\bar{v}_0)\\
&~~~+\bar{v}(E\bar{v}_0)^2+h^{-1}\bar{v}\bar{v}_0(\Delta\bar{v})
+\bar{v}\bar{v}_0(E\bar{v}_0)+h^{-1}\bar{v}^2(\Delta\bar{v}_0)
+\bar{v}^2(E\bar{v}_1)+\bar{v}\bar{v}_1(E^{-1}\bar{v})+ \bar v_0^2 \bar v,
\end{align*}
where $\bar{v}=1+h \tilde v$.
Obviously,  $\bar{A}_m$ obeys the asymptotic condition
\begin{equation}\label{4.5}
\bar{A}_m|_{\tilde{v}=0}=h^{-m}\Delta^m.
\end{equation}
The compatibility  of \eqref{dmkp-iso-lax} gives rise to
\begin{subequations}\label{dmkp-iso-com}
\begin{align}
&\bar{L}_{\bar{x}}=[\bar{A}_1,\bar{L}],\label{dmkp-iso-Lx}\\
&\bar{L}_{\bar{t}_m}=[\bar{A}_m,\bar{L}],\label{dmkp-iso-Ltm}\\
&\bar{A}_{1,\bar{t}_m}-\bar{A}_{m,\bar{x}}+[\bar{A}_1,\bar{A}_m]=0,\label{dmkp-iso-zcc}
\end{align}
\end{subequations}
Among them, Eq.\eqref{dmkp-iso-Lx} yields
\begin{subequations}\label{dmkp-v}
\begin{align}
&\bar{v}\Delta \bar{v}_0=\bar{v}_x,\label{dmkp-v0}\\
& h\bar{v}(E\bar{v}_1)-h(E^{-1}\bar{v})\bar{v}_1=h\bar{v}_{0,x}-\bar{v}(\Delta \bar{v}_0), \label{dmkp-v1}\\
& h\bar{v}(E\bar{v}_{s+1})-h(E^{-1-s}\bar{v})\bar{v}_{s+1}
=h\bar{v}_{s,x}-\bar{v}(\Delta \bar{v}_s)+\sum_{i=1}^{s}\mathrm{C}_s^{i-1}(-1)^{s-i+1}h^{i-s}\bar{v}_iE^{-s-1}\Delta^{s-i+1}\bar{v}
\label{dmkp-vs}
\end{align}
\end{subequations}
for $s\geq 1$,
from which one can `integrate' $\bar{v}_s$ with boundary condition $\bar{v}_s|_{\tilde{v}=0}=0$,
and then one has
\begin{align}\label{4.8}
\bar{v}_0=\Delta^{-1}(\ln\bar{v})_{\bar{x}},~~~
\bar{v}_1=\frac{\Delta^{-2}(\ln\bar{v})_{\bar{x}\bar{x}}
-h^{-1}\Delta^{-1}\bar{v}_{\bar{x}}}{(E^{-1}\bar{v})},~~~ \cdots.
\end{align}
Eq.\eqref{dmkp-iso-zcc} provides the zero curvature representation of the isopectral D$\Delta$mKP hierarchy,
\begin{equation}\label{4.9}
\tilde{v}_{\bar{t}_m}=\bar{K}_m(\tilde{v})
= (\bar{A}_{m,\bar{x}}-[\bar{A}_1, \bar{A}_m])\Delta^{-1},~~~m=1,2,\cdots,
\end{equation}
where the first three of them are
\begin{subequations}\label{dmkp-k}
\begin{align}
\tilde{v}_{\bar{t}_1}=\bar{K}_1(\tilde{v})&=\tilde{v}_{\bar{x}},\\
\tilde{v}_{\bar{t}_2}=\bar{K}_2(\tilde{v})
&=h^{-1}[\bar{v}(1+2\Delta^{-1})(\ln\bar{v})_{\bar{x}\bar{x}}
+\bar{v}_{\bar{x}}(1+2\Delta^{-1})(\ln\bar{v})_{\bar{x}}-2h^{-1}\bar{v}\bar{v}_{\bar{x}}],\\
\tilde{v}_{\bar{t}_3}=\bar{K}_3(\tilde{v})
&=h^{-1}[\bar{v}(3\Delta^{-2}+3\Delta^{-1}+1)(\ln\bar{v})_{\bar{x}\bar{x}\bar{x}} 
+\bar{v}_{\bar{x}}(3\Delta^{-2}+3\Delta^{-1}+1)(\ln\bar{v})_{\bar{x}\bar{x}}\nonumber\\
&~~~+\bar{v}_{\bar{x}}H_1^2
+2\bar{v} H_1 H_2  -(\bar{v} (\Delta^{-1}(\ln\bar{v})_{\bar{x}}))_{\bar x}H_1
 -\bar{v}(\Delta^{-1}(\ln\bar{v})_{\bar{x}}) H_2 \nonumber\\
&~~~ +3h^{-2}\bar{v}^2\bar{v}_{\bar{x}}-3h^{-1}(\bar{v}(1+\Delta^{-1})\bar{v}_{\bar{x}})_{\bar x} 
-3h^{-1}(\bar{v}^2\Delta^{-1}(\ln\bar{v})_{\bar{x}})_{\bar{x}}
],
\end{align}
\end{subequations}
where
\[H_1=(1+2\Delta^{-1})(\ln\bar{v})_{\bar{x}}, ~~
H_2=(1+2\Delta^{-1})(\ln\bar{v})_{\bar{x} \bar {x}}.
\]

\begin{remark}\label{Rem-4-1}
If we assume $\bar A_m$ has a form
\[\bar A_m=  \bar a_0  \Delta^m+  \bar a_1 \Delta^{m-1}+ \cdots + \bar a_{m-1} \Delta,~~
\bar A_m |_{\tilde v=0}=\Delta^m,
\]
then $\{\bar a_j\}$ can be uniquely determined from \eqref{dmkp-iso-Ltm} and it gives rise to
$\bar A_m=(\bar L^m)_{\geq 1}$.
\end{remark}

To find a master symmetry, we consider the nonisospectral Lax triad
\begin{subequations}\label{dmkp-non-lax}
\begin{align}
&\bar{L} \bar \phi=\lambda \bar \phi,~~~\lambda_{t_m}=h\lambda^m+\lambda^{m-1},\label{4.11}\\
&\bar \phi_{\bar{x}}=\bar{A}_1\bar \phi,\\
&\bar \phi_{t_m}=\bar{B}_m\bar \phi,~~~m=2, 3, \cdots,
\end{align}
\end{subequations}
where $\bar B_m$ is assumed to be a difference operator with the form
\begin{equation*}
\bar{B}_m=\sum_{j=0}^{m-1}h^{-(m-j)}\bar{b}_j\Delta^{m-j},
\end{equation*}
with the asymoptotic condition
\begin{align}\label{dmkp-non-bc}
{\bar{B}_m}|_{\tilde{v}=0}=
h^{-(m-1)}\bar{x}\Delta^m+h^{-(m-1)}(\bar{x}+nh)\Delta^{m-1}, ~~~ m \ge 2.
\end{align}
The compatibility of \eqref{dmkp-non-lax} gives rise to
\begin{subequations}\label{dmkp-non-com}
\begin{align}
&\bar{L}_{\bar{x}}=[\bar{A}_1,\bar{L}],\label{dmkp-non-Lx}\\
&\bar{L}_{\bar{t}_m}=[\bar{B}_m,\bar{L}]+h\bar{L}^m+\bar{L}^{m-1},\label{dmkp-non-Ltm}\\
&\bar{A}_{1,\bar{t}_m}-\bar{B}_{m,\bar{x}}+[\bar{A}_1,\bar{B}_m]=0.\label{dmkp-non-zcc}
\end{align}
\end{subequations}
Again, \eqref{dmkp-non-Lx} yields \eqref{dmkp-v} which expresses $\bar v_j$ in terms of $\tilde v$.
$\{\bar B_m\}$ are uniquely defined from \eqref{dmkp-non-Ltm} but
no $\bar{B}_1$ can be obtained in that approach, therefore there is no the
nonisospectral D$\Delta$mKP flow $\bar\sigma_1$ defined from \eqref{dmkp-non-com}.
The first two $ \bar B_m $ are
\begin{align*}
&\bar{B}_2=h\bar{x}\bar{A}_2+(\bar{x}+nh)\bar{A}_1,\\
&\bar{B}_3=h\bar{x}\bar{A}_3+(\bar{x}+nh)\bar{A}_2
+\bar{v}(\Delta^{-1}\bar{v}_0)\Delta-h^{-1}\bar{v}(\Delta^{-1}\bar{v})\Delta
+h^{-1}n\bar{v}\Delta.
\end{align*}
where $\bar A_j$ are given in \eqref{Aj}.
Then, \eqref{dmkp-non-zcc} defines the  nonisopectral D$\Delta$mKP  hierarchy
\begin{equation}\label{4.16}
\tilde{v}_{\bar{t}_m}=\bar{\sigma}_m(\tilde{v})= (\bar{B}_{m,\bar{x}}+[\bar{B}_m,\bar{A}_1])\Delta^{-1},
~~  m=2,3,\cdots,
\end{equation}
in which  the first two equations are
\begin{align*}
\tilde{v}_{\bar{t}_2}&=\bar{\sigma}_2(\tilde{v})=h\bar{x}\bar{K}_2
+(\bar{x}+nh)\bar{K}_1+\bar{v}(1+2\Delta^{-1})(\ln\bar{v})_{\bar{x}}+ h^{-1} \bar{v}- h^{-1} {\bar{v}}^2.\\
\tilde{v}_{\bar{t}_3}&=\bar{\sigma}_3(\tilde{v})  \\
& =h\bar{x}\bar{K}_3
+(\bar{x}+nh)\bar{K}_2+h^{-2}\bar{v}(E\bar{v})(\Delta^2\bar{v})
+2h^{-1}\bar{v}(E\bar{v})(E\Delta\bar{v}_0)+\bar{v}(E\bar{v})(E^2\bar{v}_1)  \\
&~~~+h^{-2}\bar{v}(\Delta\bar{v})^2+2h^{-1}\bar{v}(\Delta\bar{v})(E\bar{v}_0)
+\bar{v}(E\bar{v}_0)^2+h^{-1}\bar{v}\bar{v}_0(\Delta\bar{v})
+\bar{v}\bar{v}_0(E\bar{v}_0) \\
&~~~+h^{-1}\bar{v}^2(\Delta\bar{v}_0)+\bar{v}^2(E\bar{v}_1)
+\bar{v}\bar{v}_1(E^{-1}\bar{v})+ \bar v^2_0 \bar v +h^{-2}\bar{v}(\Delta \bar{v})
+h^{-1}\bar{v}(E\bar{v}_0)  \\
&~~~+h^{-1}\bar{v}\bar{v}_0 
+(\bar{v}(\Delta^{-1}\bar{v}_{0}))_{\bar{x}} 
-h^{-1}(\bar{v}(\Delta^{-1}\bar{v}))_{\bar{x}} 
 +h^{-1}n\bar{v}_{\bar{x}}-h^{-1}\bar{v}(E\bar{v})\bar{v}_0,
\end{align*}
where $\bar{v}_0$ and $\bar{v}_1$ are given in \eqref{4.8} and
$\bar K_j$ are given in \eqref{dmkp-k}.
In addition, one can find that
\begin{equation}
\bar{K}_m|_{\tilde{v}=0}=0,~~~ \bar{\sigma}_s|_{\tilde{v}=0}=0,~~~(m=1,2,\cdots, ~~s=2,3,\cdots).
\end{equation}

Eqs. \eqref{4.9} and \eqref{4.16} are the zero curvature representations of the
D$\Delta$mKP flows $\{\bar{K}_m,\bar{\sigma}_s\}_{m \ge 1, s \ge 2}$,
which can be used to determine algebraic structures of these flows
(e.g. \cite{CXZ-2003,CZ-JMP-1996,FHTZ-2013,MF-JMP-1999}).
To proceed, one should have a zero curvature equation that allows unique solutions.

\begin{lemma}\label{Lem-1}
Suppose  $\bar{X}=\bar{X}(\tilde v) \in \bar{S}[\tilde v]$ and $\bar N$ is a difference operator living on
$\bar{S}[\tilde v]$, with a form
\begin{equation*}
\bar{N}=\bar{a}_0\Delta^m+\bar{a}_1\Delta^{m-1}+\cdots+\bar{a}_{m-1}\Delta,
\end{equation*}
where $\bar{a}_j=\bar{a}_j(\tilde v) \in \bar{S}[\tilde v]$ and it follows that $\bar{N}|_{\tilde{v}=0}=0$.
Then, the following equation
\begin{equation}\label{4.20}
\bar{X}-\bar{N}_{\bar{x}}\Delta^{-1}+[\bar{A}_1,\bar{N}]\Delta^{-1}=0
\end{equation}
has only  zero solutions $\bar{X}=0$ and  $\bar{N}=0$,
where $\bar{A}_1$ is given as \eqref{A1}.
\end{lemma}

In fact, since $\bar v =1+h \tilde v \to 1$ as $|n|\to \infty$,
substituting $\bar{N}$ into \eqref{4.20} one can successfully find
$\bar{a}_0=0$, $\bar{a}_1=0$, $\cdots$, $\bar{a}_{m-1}=0$
and finally $\bar N=0$ and $\bar X=0$.

Then we have the following.
\begin{theorem}\label{Th-3}
The D$\Delta$mKP flows $\{\bar{K}_m,\bar{\sigma}_s\}_{m \ge 1, s \ge 2}$ compose a Lie algebra with the following structure
\begin{subequations}\label{K-S-alg}
\begin{align}
&\llbracket \bar{K}_m,\bar{K}_s \rrbracket=0,\\
&\llbracket \bar{K}_m,\bar{\sigma}_s \rrbracket=m (h \bar{K}_{m+s-1}+\bar{K}_{m+s-2}), \label{K-S-alg-2}\\
&\llbracket \bar{\sigma}_m,\bar{\sigma}_s \rrbracket=(m-s) (h\bar{\sigma}_{m+s-1}+\bar{\sigma}_{m+s-2}).
\end{align}
\end{subequations}
\end{theorem}

\begin{proof}
First, for the D$\Delta$mKP flows $\{\bar{K}_m,\bar{\sigma}_n\}_{m \ge 1,n \ge 2}$,
using their zero curvature representations, we can find that
\begin{subequations}
\begin{align}
&\llbracket \bar{K}_m,\bar{K}_s\rrbracket=(\langle \bar{A}_m,\bar{A}_s\rangle_y
+[\langle \bar{A}_m,\bar{A}_s\rangle,\bar{A}_2])\Delta^{-1},\\
&\llbracket \bar{K}_m,\bar{\sigma}_s\rrbracket= (\langle \bar{A}_m,\bar{B}_s\rangle_y
+[\langle \bar{A}_m,\bar{B}_s\rangle,\bar{A}_2])\Delta^{-1},\label{km,sn}\\
&\llbracket \bar{\sigma}_m,\bar{\sigma}_s\rrbracket= (\langle \bar{B}_m,\bar{B}_s\rangle_y
+[\langle \bar{B}_m,\bar{B}_s\rangle,\bar{A}_2])\Delta^{-1},
\end{align}
\end{subequations}
where
\begin{align*}
&\langle \bar{A}_m,\bar{A}_s\rangle=\bar{A}_m'[\bar{K}_s]-\bar{A}_s'[\bar{K}_m]+[\bar{A}_m,\bar{A}_s],\\
&\langle \bar{A}_m,\bar{B}_s\rangle=\bar{A}_m'[\bar{\sigma}_s]
- \bar{B}_s' [\bar{K}_m]+[\bar{A}_m,\bar{B}_s],\\
&\langle \bar{B}_m,\bar{B}_s\rangle=\bar{B}_m'[\bar{\sigma}_s]-\bar{B}_s'[\bar{\sigma}_m]
+[\bar{B}_m,\bar{B}_s].
\end{align*}
Let us take \eqref{km,sn} as an example to explain the above relations.
From the zero curvature representation \eqref{4.9} we have
\begin{subequations}
\begin{align*}
\bar{K}_m'(\tilde v)[\bar{\sigma}_s]&=((\bar{A}_{m,\bar{x}}
+[\bar{A}_m,\bar{A}_1])\Delta^{-1})'[\bar{\sigma}_s]\\
&=\bigl(\bar{A}_m'[\bar{\sigma}_s])_{\bar{x}}+[\bar{A}_m'[\bar{\sigma}_s],\bar{A}_1]
+[\bar{A}_m,\bar{A}_1'[\bar{\sigma}_s]]\bigr) \Delta^{-1}\\
&=\bigl( (\bar{A}_m'[\bar{\sigma}_s])_{\bar{x}}
+[\bar{A}_m'[\bar{\sigma}_s],\bar{A}_1]+[\bar{A}_m,h^{-1}\bar{\sigma}_s\Delta]]\bigr) \Delta^{-1}\\
&=\bigl( (\bar{A}_m'[\bar{\sigma}_s])_{\bar{x}}+[\bar{A}_m'[\bar{\sigma}_s],\bar{A}_1]
+[\bar{A}_m,\bar{B}_{s,\bar{x}}]+[\bar{A}_m,[\bar{B}_s,\bar{A}_1]]\bigr) \Delta^{-1}.
\end{align*}
\end{subequations}
Similarly, from \eqref{4.16} we have
\begin{equation*}
\bar{\sigma}_s'(\tilde v)[\bar{K}_m]=\bigl( (\bar{B}_s'[\bar{K}_m])_{\bar{x}}
+[\bar{B}_s'[\bar{K}_m],\bar{A}_1]
+[\bar{B}_s,\bar{A}_{m,\bar{x}}]+[\bar{B}_s,[\bar{A}_m,\bar{A}_1]]\bigr) \Delta^{-1}.
\end{equation*}
Subtracting $\bar{\sigma}_s'[\bar{K}_m]$ from $\bar{K}_m'[\bar{\sigma}_s]$
and making use of the Jacobi identity of $[\cdot,\cdot]$, we obtain \eqref{km,sn}.

Next, noticing that the following asymptotic conditions
\begin{align*}
&\langle \bar{A}_m,\bar{A}_s\rangle|_{\bar{v}=1}=0,\\
&\langle \bar{A}_m,\bar{B}_s\rangle|_{\bar{v}=1}=h^{-(m+s-2)}m(\Delta^{m+s-1}+\Delta^{m+s-2}), \\
&\langle \bar{B}_m,\bar{B}_s\rangle|_{\bar{v}=1}=h^{-(m+s-3)}(m-s)(\bar{x}\Delta^{m+s-1}
+(2\bar{x}+hn)\Delta^{m+s-2}+(\bar{x}+hn)\Delta^{m+s-3})
\end{align*}
hold, comparing them with \eqref{4.5} and \eqref{dmkp-non-bc},
and using Lemma \ref{Lem-1},
we obtain the algebraic relations \eqref{K-S-alg}.

\end{proof}

With these structures we have the following.

\begin{theorem}\label{Th-4-2}
Each equation
\begin{equation}
\tilde v_{t_k}=\bar K_k
\end{equation}
in the isospectral D$\Delta$mKP hierarchy \eqref{4.9} has two sets of symmetries:
\begin{equation}  \label{4.26}
  \{ \bar{K}_s \}_{s \geq 1}, \qquad
  \{ \bar{\tau}_m^k= k \  \bar t_k  (h \bar{K}_{m+k-1}+\bar{K}_{m+k-2})+\bar{\sigma}_m \}_{m \geq 2},
\end{equation}
and they compose a Lie algebra with structure
\begin{subequations}\label{4.27}
\begin{align}
    &  \llbr \bar{K}_s, \bar{K}_m \rrbr  =      0,                \\
    &  \llbr \bar{K}_s, \bar{\tau}_m^k \rrbr  =  s (h \bar{K}_{m+s-1}+\bar{K}_{m+s-2}),                        \\
    &  \llbr \bar{\tau}_s^k, \bar{\tau}_m^k \rrbr   =   (s-m) (h \bar{\tau}^k_{m+s-1}+\bar{\tau}^k_{m+s-2}).
\end{align}
\end{subequations}
\end{theorem}

\subsection{Squared eigenfunction symmetry constraint}\label{sec-4-2}

\subsubsection{Symmetry and constraint}\label{sec-4-2-1}

Consider the triplet \eqref{dmkp-iso-lax} and its adjoint form
with eigenfunctions $\Phi$ and $\Phi^*$, respectively,
we introduce the adjoint form
\begin{subequations}\label{dmkp-iso-lax-adj}
\begin{align}
&\bar{L} \Phi=\lambda \Phi, ~~~~~ \bar{L}^* \Phi^*=\lambda \Phi^*,\label{dmkp-sp-adj}\\
& \Phi_{\bar{x}}=\bar{A}_1 \Phi, ~~~~ \Phi^*_{\bar{x}}=-\bar{A}^*_1 \Phi^*, \label{dmkp-iso-A1-adj}\\
& \Phi_{\bar{t}_s}=\bar{A}_s \Phi, ~~~ \Phi^*_{\bar{t}_s}=-\bar{A}^*_s \Phi^*,~~s=1, 2, \cdots, \label{dmkp-flow-lax-adj}
\end{align}
\end{subequations}
where $\bar{L}^*$ and $\bar{A}^*_s$ are the adjoint operators of $\bar L$ and $\bar{A}_s$, respectively,
defined through the inner product \eqref{2.8}.
It can be proved that  \cite{CZZ-2021}
\begin{equation}
\omega=h(\Phi E \Delta^{-1} \Phi^*)_{\bar{x}}=-\sum^{\infty}_{s=1} \bar K_s (\tilde v),
\end{equation}
i.e., $\omega$ is the squared eigenfunction symmetry of the D$\Delta$mKP hierarchy.
Note that in \cite{CZZ-2021}, by the constraint
\begin{equation}\label{sym-con-RT}
	\bar{v}=\Phi E \Delta^{-1} \Phi^*,
\end{equation}
the D$\Delta$mKP system composed by (with $\bar x=\bar{t}_1$)
\begin{subequations}\label{dmkp-iso-lax-adj-kn}
\begin{align}
&\bar{L}  \Phi =\lambda \Phi, \label{dmkp-iso-L-adj-kn}\\
&\Phi_{\bar{t}_s}=-\bar{A}_s \Phi,~~
\Phi^*_{\bar{t}_s}=-\bar{A}^*_s \Phi^*,~~s=1, 2, \cdots \label{dmkp-flow-lax-adj-kn}
\end{align}
\end{subequations}
gives rise to the  the spectral problem and positive hierarchy of the relativistic Toda,
not those of the DNLS as we expect (cf. in the continuous case \cite{Chen-mKP-2002}).

\subsubsection{The discrete CLL spectral problem}\label{sec-4-2-2}

In the following, instead of \eqref{sym-con-RT}, we introduce another constraint
 (note that $\bar v=1+ h\tilde v$)
\begin{equation}\label{sym-con}
	\tilde{v}=h \Phi E \Delta^{-1} \Phi^*=a_nb_n,~~~ a_n=\Phi,~~~b_n=h E\Delta^{-1}\Phi^*,
\end{equation}
under which, first, one can rewrite $\bar L$ in terms of $a_n$ and $b_n$.

\begin{proposition}\label{P-4-1}
For the two triplets in \eqref{dmkp-iso-lax-adj},
with the constraint \eqref{sym-con}, the pseudo-difference operator $\bar L$  \eqref{dmkp-L}
can be written into the form
\begin{equation}\label{dmkp-L-com}
\bar{L}=h^{-1}(1+ha_n b_n)\Delta+ a_n \Delta^{-1} b_n \Delta.
\end{equation}
\end{proposition}

\begin{proof}
In light of formula \eqref{2.7} with $s=-1$,
we only need to prove that all the $\{\bar{v}_s\}^\infty_{s=0}$
defined by \eqref{dmkp-vs} can be expressed in terms of $a_n$ and $b_n$ as the following:
\begin{equation}\label{vs-compact}
\bar{v}_s=(-1)^sh^{-s}a_nE^{-s-1}\Delta^sb_n,~~~s=0, 1, 2, \cdots.
\end{equation}
Then \eqref{dmkp-L-com} follows.

The proof of \eqref{vs-compact} is similar to Proposition 3 in \cite{CZZ-2021}.
Below we sketch main steps and leave the details in Appendix \ref{App-1}.
First, it follows from \eqref{dmkp-iso-A1-adj} and \eqref{sym-con} that
\begin{equation}\label{4.35}
	a_{n,\bar{x}}=h^{-1}(1+ ha_nb_n)(a_{n+1}-a_n),~~b_{n,\bar{x}}=h^{-1}(1+ha_nb_n)(b_n-b_{n-1}),
\end{equation}
which will be used to eliminate derivatives in \eqref{dmkp-v},
in particular, from \eqref{dmkp-v0} and \eqref{dmkp-v1}  we have
\begin{equation*}
\bar{v}_0=a_n b_{n-1},~~\bar{v}_1=-h^{-1}a_nE^{-2}\Delta b_n.
\end{equation*}
Then, to make use of mathematical induction, we assume the following expression holds for
$i=0, 1, 2, \cdots, s$, i.e.
\begin{equation}\label{4.37}
\bar{v}_i=(-1)^i h^{-i }a_nE^{-i-1}\Delta^i b_n,~~~i=0, 1, 2, \cdots, s.
\end{equation}
Then, for $s+1$, the right-hand side of \eqref{dmkp-vs} gives rise to (refer to Appendix A for details)
\begin{equation}\label{vs-rhs}
(-1)^{s+1}h^{-s}[ (1+ha_n b_n)E(a_nE^{-s-2}\Delta^{s+1}b_n)
 - (E^{-s-1}(1+ ha_n b_n))(a_nE^{-s-2}\Delta^{s+1}b_n)].
\end{equation}
Compared with the left-hand side of \eqref{dmkp-vs}, we can find
\begin{equation*}
\bar{v}_{s+1}=(-1)^{s+1}h^{-s-1}a_nE^{-s-2}\Delta^{s+1}b_n.
\end{equation*}
The proof is completed.

\end{proof}

\begin{theorem}\label{Th-5}
With the constraint \eqref{sym-con}, the spectral problem \eqref{4.11} gives rise to the
discrete CLL spectral problem
\begin{equation}\label{delta-psi}
\left(\begin{array}{c} \bar \psi_{1,n+1}\\ \bar \psi_{2,n+1} \end{array}\right)=
\left(\begin{array}{cc} -h\eta^2+1+h a_nb_n & h \eta a_n\\
h \eta b_n & 1 \end{array}\right)
\left(\begin{array}{c} \bar \psi_{1,n}\\ \bar \psi_{2,n} \end{array}\right).
\end{equation}
\end{theorem}

\begin{proof}
In light of Proposition \ref{P-4-1}, the spectral problem \eqref{4.11} is written as
\begin{equation}\label{com-sp}
\bar L \bar \phi =[h^{-1}(1+h a_n b_n )\Delta+ a_n \Delta^{-1} b_n \Delta ]\bar \phi=\lambda \bar \phi.
\end{equation}
Introducing $\varphi_{1,n}$ and $\varphi_{2,n}$ by
\begin{equation}
\lambda=-\eta^2,~~
\bar \varphi_{1,n}=\bar \phi,~~ \eta \Delta \bar \varphi_{2,n}=-b_n\Delta \bar \varphi_{1,n},
\end{equation}
\eqref{com-sp} yields
\begin{subequations}\label{delta-phi}
\begin{align}
&(1+h a_nb_n)\Delta \bar \varphi_{1,n}=-h\eta^2 \bar \varphi_{1,n}+h \eta a_n \bar \varphi_{2,n},\\
&(1+h a_nb_n)\Delta \bar \varphi_{2,n}=h  \eta b_n \bar \varphi_{1,n}-h a_n b_n \bar \varphi_{2,n},
\end{align}
\end{subequations}
i.e.
\begin{align*}
&(1+h a_nb_n)\bar \varphi_{1,n+1}=(-h\eta^2+1+h a_nb_n)\bar\varphi_{1,n}+h \eta a_n \bar\varphi_{2,n},\\
&(1+h a_nb_n)\bar \varphi_{2,n+1}=h \eta b_n \bar\varphi_{1,n}+\bar\varphi_{2,n}.
\end{align*}
Next, we introduce a gauge transformation
\begin{equation*}
\left(\begin{array}{c} \bar\psi_{1,n}\\ \bar\psi_{2,n} \end{array}\right)=
\left(\begin{array}{cc} (1+h a_nb_n)^n & 0 \\ 0 & (1+h a_nb_n)^n \end{array}\right)
\left(\begin{array}{c} \bar\varphi_{1,n}\\ \bar\varphi_{2,n} \end{array}\right),
\end{equation*}
and immediately we arrive at the spectral problem \eqref{delta-psi}.
If we further introduce
\begin{equation*}
\left(\begin{array}{c} \bar\psi_{1,n}\\ \bar\psi_{2,n} \end{array}\right)=
\left(\begin{array}{cc} 1 & 0\\   0 &1/\eta \end{array}\right)
\left(\begin{array}{c}   \psi_{1,n}\\  \psi_{2,n} \end{array}\right),
\end{equation*}
\eqref{delta-psi} gives rise to
\begin{equation}\label{sd-CLL-sp}
\left(\begin{array}{c}  \psi_{1,n+1}\\  \psi_{2,n+1} \end{array}\right)=
\left(\begin{array}{cc} -h\eta^2+1+h a_nb_n & h   a_n\\
h \eta^2 b_n & 1 \end{array}\right)
\left(\begin{array}{c}  \psi_{1,n}\\  \psi_{2,n} \end{array}\right),
\end{equation}
which is known as the  discrete CLL spectral problem \cite{DJM-1983},
see also \cite{T-JPA-2002,KMW-TMP-2013}.

\end{proof}

We will introduce more about the discrete CLL spectral problem \eqref{delta-psi} and \eqref{sd-CLL-sp}
in Appendix \ref{App-2}.

\subsubsection{The differential-difference CLL hiararchy}\label{sec-4-2-3}

Next, let us look at  \eqref{dmkp-flow-lax-adj-kn} and reveal their explicit forms
under the constraint \eqref{sym-con} and the $(a_n,b_n)$ formulation.
The procedure is similar to \cite{CZZ-2021} but we need examine all details.
First, we have the following.

\begin{proposition}\label{P-4-2}
For $\bar L$ given in \eqref{dmkp-L-com}, $A_s$ obeys the following two recursive relations:
\begin{subequations}\label{dmkp-as-recursive}
\begin{align}
&\bar{A}_{s+1}=\bar{L}\bar{A}_s+h^{-1}(ha_nb_n+1)(E(\bar{L}^s)_0)\Delta
-a_n\Delta^{-1}(E\Delta^{-1}\bar{A}_s^*E^{-1}\Delta b_n)\Delta,
\label{dmkp-as-recursive1} \\
&\bar{A}_{s+1}=\bar{A}_s\bar{L}+h^{-1}(\bar{L}^s)_0(ha_n b_n+1)
\Delta-(\bar{A}_s a_n)\Delta^{-1}b_n \Delta,\label{dmkp-as-recursive2}
\end{align}
\end{subequations}
where $(\bar{L}^s)_0$ stands for  the constant term of  $\bar{L}^s$ with respect to $\Delta$.
\end{proposition}

\begin{proof}
Using the identities \cite{CZZ-2021}
\begin{equation*}
(\bar{A}_sa_n\Delta^{-1}b_n\Delta)_{\le 0}=(\bar{A}_s a_n)\Delta^{-1}b_n,
~~\Delta^{-1}b_n\Delta \bar{A}_s=\Delta^{-1}(E\Delta^{-1}\bar{A}_s^*E^{-1}\Delta b_n)\Delta,
\end{equation*}
and noticing that $\bar{A}_{s+1}=(\bar{L}\bar{L}^s)_{\ge 1}$ and
$\bar{A}_{s+1}=(\bar{L}^s\bar{L})_{\ge 1}$,
one can obtain \eqref{dmkp-as-recursive1} and \eqref{dmkp-as-recursive2} respectively.
\end{proof}

\noindent
Next, we have the following.

\begin{theorem}\label{Th-6}
Equation \eqref{dmkp-flow-lax-adj-kn} gives rise to the recursive hierarchy
\begin{equation}\label{recursive-ab}
\left(\begin{array}{c}a_n\\b_n\end{array}\right)_{\bar{t}_{s+1}}
=\bar T \left(\begin{array}{c}a_n\\b_n\end{array}\right)_{\bar{t}_{s}},
\end{equation}
where the initial equation reads
\begin{equation}\label{int}
a_{n,\bar{t}_1}=h^{-1}(1+ h a_nb_n)(a_{n+1}-a_n),~~b_{n,\bar{t}_1}=h^{-1}(1+ h a_nb_n)(b_n-b_{n-1}),
\end{equation}
and the recursion operator is
\begin{equation}
\bar T=\left(\begin{array}{cc}\bar T_{11} & \bar T_{12} \\ \bar T_{21} & \bar T_{22}\end{array}\right)
\end{equation}
with entries
\begin{equation}\label{T-ij}
\begin{array}{l}
\bar T_{11}=h^{-1}(ha_nb_n\Delta+\Delta+ha_n\Delta^{-1}b_n\Delta)
+h^{-1}(ha_nb_n+1)(\Delta a_n)E\Delta^{-1}\frac{hb_n}{ha_nb_n+1},\\
\bar T_{12}=h^{-1}(ha_nb_n+1)(\Delta a_n)E\Delta^{-1}\frac{ha_n}{ha_n b_n+1}
-a_n \Delta^{-1}(\Delta a_n),\\
\bar T_{21}=h^{-1}(ha_nb_n+1)(E^{-1}\Delta b_n)
\Delta^{-1}\frac{hb_n}{ha_nb_n+1}+b_nE\Delta^{-1}(E\Delta b_n),\\
\bar T_{22}=h^{-1}(-ha_nb_nE^{-1}\Delta-E^{-1}\Delta
+hb_nE\Delta^{-1}a_nE^{-1}\Delta)\\
~~~~~~~~~ +h^{-1}(ha_nb_n+1)(E^{-1}\Delta b_n)\Delta^{-1}\frac{ha_n}{ha_n b_n+1}.
\end{array}
\end{equation}
\end{theorem}

\begin{proof}
The initial equation \eqref{int} is obvious.
Next, recalling the residue form of the D$\Delta$mKP hierarchy \cite{CZZ-2021},
\begin{equation}
\tilde{v}_{\bar{t}_m}=\bar{K}_m=h^{-1}\bar v \Delta\mathop{\mathrm{Res}}\limits_{\Delta}(L^m\Delta^{-1}),
\end{equation}
one has
\begin{equation*}
(\bar{L}^s)_0=\Delta^{-1}(\ln\bar{v})_{\bar{t}_s}=\Delta^{-1}(\ln(ha_nb_n+1))_{\bar{t}_s}
=\Delta^{-1}\left(\frac{hb_n}{ha_nb_n+1}a_{n,\bar{t}_s}+\frac{ha_n}{ha_nb_n+1}b_{n,\bar{t}_s}\right).
\end{equation*}
Then using \eqref{dmkp-as-recursive1} we have
\begin{align}
a_{n,\bar{t}_{s+1}}&=\bar{A}_{s+1}a_n \nonumber \\
&=\bar{L} a_{n,\bar{t}_{s}}
+h^{-1}(ha_nb_n+1)(E(\bar{L}^s)_0)\Delta a_n+a_n\Delta^{-1}(\Delta a_n) b_{n,t_{s}}, \label{a-L0}
\end{align}
where we have used
\[a_{n,\bar{t}_{s}}=\bar{A}_{s}a_n, ~~ b_{n,t_{s}}=- E\Delta^{-1}\bar{A}_s^*E^{-1}\Delta b_n.\]
Then, substituting $(\bar{L}^s)_0$ into \eqref{a-L0} we have
\[a_{n,\bar{t}_{s+1}}=\bar T_{11} a_{n,\bar{t}_{s}} + \bar T_{12} b_{n,\bar{t}_{s}},\]
where $\bar T_{11}$ and $\bar T_{12}$ are given in \eqref{T-ij}.
Similarly, one can prove that
\[b_{n,\bar{t}_{s+1}}=\bar T_{21} a_{n,\bar{t}_{s}} + \bar T_{22} b_{n,\bar{t}_{s}}.\]
Thus, the proof is completed.

\end{proof}

In the next section, we will show a continuum limit in which
\eqref{recursive-ab} recovers  the CLL  hierarchy.

\section{Continuum limits of the D$\Delta$mKP}\label{sec-5}

Our purpose of this section is to implement continuum limits on the integrable structures of the D$\Delta$mKP
and match them with those of the continuous case.
The technique of implementing the continuum limits is similar to that we have used in \cite{FHTZ-2013}
for investigating the D$\Delta$KP,
where we introduced a notion called  ``degree'' as a practical and effective tool to
figure out the leading terms in the continuum limits.
In the following we will first present a unified scheme of the continuum limits
and calculate degree of each object of the D$\Delta$mKP.
After that the continuum limits will be investigated.
We will skip some details, for which one can refer to the analogues in \cite{FHTZ-2013}.
We will see that all the results in Section \ref{sec-3} for the mKP
will be recovered from those of the  D$\Delta$mKP in the continuum limits.


\subsection{Scheme of the continuum limits}\label{sec-5-1}

Our plan for the continuum limit is as the following (cf.\cite{FHTZ-2013} for the D$\Delta$KP case).
\begin{itemize}
\item Let $n \rightarrow \infty$ and $h \rightarrow 0$ simultaneously such that $nh$ is finite.
\item Introduce an auxiliary continuous variable
\begin{equation}
\tau=nh,
\end{equation}
which maps $\bar f(n+j)$ to $f(\tau+jh)$.
\item Define coordinates relation
\begin{equation}
x=\bar{x}+\tau,~~~y=-\frac{1}{2}h\tau,~~~t_m=\bar{t}_m,
\end{equation}
under which one has
\begin{equation}
\partial_{\bar{x}}=\partial_x,~~~\partial_{\tau}=\partial_x-\frac{1}{2}h\partial_y,~~~
\partial_{\bar{t}_m}=\partial_{t_m}.
\end{equation}
\item Define functions relation
\begin{subequations}\label{com-lim-fcts}
\begin{align}
&\bar{v}(n,\bar{x},\bar{t}_m)=1+h\tilde{v}(n,\bar{x},\bar{t}_m)=1+h v (x,y,t_m),\label{com-lim-v}\\
& \bar \phi (n,\bar{x},\bar{t}_m)=\phi (x,y,t_m).
\end{align}
\end{subequations}
\end{itemize}

Base on the above scheme one can quickly find
\begin{align}
\Delta &=h \partial_{\tau}+\frac{1}{2!}h^2\partial_{\tau}^2+\frac{1}{3!}h^3\partial_{\tau}^3+O(h^4) \nonumber\\
&=h\partial_x+\frac{h^2}{2}(\partial_x^2-\partial_y)+\frac{h^3}{6}(\partial_x^3-3\partial_x\partial_y)+O(h^4).
\label{5.5}
\end{align}
Consequently,
\begin{align}
\Delta^{-1}&=h^{-1}\partial_x^{-1}+\frac{1}{2}(\partial_x^{-2}\partial_y^{-1}-1)+O(h),
\end{align}
and in general,
\begin{equation}
\Delta^j=h^{j}\partial^j_x+O(h^{j+1}), ~~ j\in \mathbb{Z}.
\end{equation}

\subsection{Degrees}\label{sec-5-2}

By degree we mean (see Definition 5.1 in \cite{FHTZ-2013})
in the scheme presented in Section \ref{sec-5-1}
the order of the leading term of a function $\bar f(n, \bar x, \bar t_m)$ (or operator $\bar P(\tilde v, \Delta)$)
after  expanded in a series   of $h$, denoted by deg $\bar f$ (or deg $\bar P$).

Let us first examine the degrees of $\{\bar{v}_j\}$. From \eqref{dmkp-v} we find
\begin{equation*}
\bar{v}_0 =\Delta^{-1}(\ln (1+h\tilde{v}))_{\bar{x}}=v+O(h),
\end{equation*}
and
\begin{equation*}
\bar{v}_1=-\frac{1}{2}\tilde{v}^2+\frac{1}{2}(\partial_x^{-1}\partial_y-\partial_x)\tilde{v}+O(h)
=v_1+O(h).
\end{equation*}
For general $j$, if
we assume $\mathrm{deg}\, \bar v_j=0$ for $j=1,2,\cdots, s$,
then, by analyzing both sides of \eqref{dmkp-vs},
we can find the degree of $\bar v_{s+1}$ must be zero too.
Thus, we have
\[ \mathrm{deg}\, \bar v_j=0\]
and then we may assume
\[\bar v_j(n, \bar x, \bar t_m)=v_j(x,y,t_m)+ O(h),~~ j=1,2,\cdots.\]
Thus, for the pseudo-difference operator $\bar L$  given in \eqref{dmkp-L} we have
\begin{equation}
\mathrm{deg}\, \bar L=0,~~ \bar{L}=L+O(h),
\end{equation}
and consequently,
\begin{equation}\label{5.12a}
\bar L \bar \phi - \lambda \bar \phi = (L\phi-\lambda \phi )+ O(h).
\end{equation}
In addition, for $\bar A_1$,  using \eqref{5.5}, we find
\begin{equation}\label{A1-A2}
\bar A_1 = A_1 +\frac{h}{2}(A_2 - \partial_y) + O(h^2),
\end{equation}
and it then follows that
\begin{align}\label{5.12b}
\bar \phi_{\bar x}- \bar A_1 \bar \phi = \frac{h}{2}(\phi_y -A_2\phi)+ O(h^2)
\end{align}
and
\begin{equation*}
 \bar L_{\bar x}-[\bar A_1, \bar L]= \frac{h}{2}(L_y -[A_2, L])+ O(h^2).
\end{equation*}
Besides, for the first few flows of $\{\bar K_j\}$ and $\{\bar\sigma_j\}$, one can also check that
\begin{align*}
&\mathrm{deg}\,\bar K_j =0,~~  \bar K_j=K_j+O(h),~~ (j=1,2,3),\\
& \mathrm{deg}\,\bar \sigma_j =0,~~  \bar \sigma_j=\sigma_j+O(h),~~ (j=2,3).
\end{align*}

In the following, in order to obtain degrees of more elements,
we recall some rules for calculating degrees developed in \cite{FHTZ-2013}.
Let us list them below.

\begin{proposition}\label{prop:deg-1}
For  functions (or operators) $\bar f(\tilde v)$ and $\bar g(\tilde v)$, it holds that
\begin{subequations}\label{deg-1}
\begin{align}
&\deg\bar f\cdot\bar g=\deg\bar f+\deg\bar g,\label{deg:fg}\\
&\deg(\bar f+\bar g)\geq\min\{\deg\bar f,\deg\bar g\}. \label{deg:f+g}
\end{align}
\end{subequations}
\end{proposition}

\begin{proposition}\label{prop:deg-2}
For functions $\bar f(\tilde v),\bar g(\tilde v) \in \bar S[\tilde v]$
and $f(v), g(v) \in S[v]$,  suppose that there are relations in the continuum limit:
\begin{align*}
\bar f(\tilde v)=f(v)h^i+O(h^{i+1}),\quad \bar g(\tilde v)=g(v)h^j+O(h^{j+1}).
\end{align*}
It then holds that
\begin{align}
&\llbracket \bar f(\tilde v),\bar g(\tilde v)\rrbracket_{\tilde v}
=\llbracket f(v),g(v)\rrbracket_{v}\, h^{i+j-1}+O(h^{i+j}),\label{f-g-deg1}\\
&\deg\llbracket \bar f(\tilde v),\bar g(\tilde v)\rrbracket_{\tilde v}
\geq\deg\b f(\b u)+\deg\b g(\b u)-1,\label{f-g-deg2}
\end{align}
where the subscripts $\tilde v$ and $v$ indicate the G\^ateaux derivatives in $\llbracket \cdot,\,\cdot \rrbracket$
are defined with respect to $\tilde v$ and $v$, respectively.
\end{proposition}

In addition, similar to Lemma 5.1, Lemma 5.2 and Proposition 5.6 in \cite{FHTZ-2013},
for the D$\Delta$mKP system, we have the following.

\begin{lemma}\label{Lem-5-1}
For the difference operator
\begin{equation}\label{W-bar}
\bar{W}_m=\sum_{j=0}^{m-1}\bar{w}_j(\tilde{v})\Delta^{m-j},~~~\mathrm{with}~
\bar{w}_j(\tilde{v})\in \bar S[\tilde v],~ ~ j=0,1,m-1,
\end{equation}
if $\bar{W}_m$ satisfies $\left[\bar{W}_m, \bar{L} \right]=0$, then $\bar{W}_m=0$.
\end{lemma}
\begin{proof}
Writing the left-hand side of  $\left[\bar{W}_m, \bar{L} \right]=0$ in terms of  the power of $\Delta$,
from the highest order term, we find
\begin{equation*}
(E\bar{w}_0)=\frac{(E^m\bar{v})}{\bar{v}}\bar{w}_0.
\end{equation*}
This is a first order ordinary difference equation with respect to $\bar{w}_0$.
It allows a general solution
\[\bar{w}_0 =c\prod_{j=0}^{m-1}(E^j\bar{v}),\]
and since $\bar{w}_0\to 0$ and $\bar{v}\to 1$ as $|n|\to \infty$,
the constant $c$ has to be zero, and thus $\bar{w}_0=0$.
In the same way we can successfully work out $\bar{w}_j=0$ for $j=1,2,\cdots, m-1$.

\end{proof}

Similarly, the following holds.
\begin{lemma}\label{Lem-5-2}
For the differential operator
\begin{equation*}
W_m=\sum_{j=0}^{m-1}w_j(v)\partial^{m-j},~~~\mathrm{with}~w_j(v)\in S[v],~~  j=0,1,2,\cdots, m-1,
\end{equation*}
if $W_m$ satisfies $[W, L]=0$, then $W_m=0$.
\end{lemma}

With these two lemmas, for the difference operator $\bar{W}_m$ given in \eqref{W-bar}, we can also prove that
(cf. Proposition 5.6 in \cite{FHTZ-2013})
\begin{equation}\label{deg:W-L}
\deg [\bar{W}_m, \bar L]=0.
\end{equation}
The above two lemmas and relation \eqref{deg:W-L} are important to  achieve the relations between
$\bar A_m, \bar B_m$ and $A_m, B_m$ in the continuous limit.
Below we present these relations, which can be proved via
similar procedures as for Proposition 5.7 and 5.8 in \cite{FHTZ-2013}.

\begin{proposition}
In the continuum limit, we have
\begin{subequations}\label{5.22}
\begin{align}
&{\rm{deg}}~\bar{A}_m=0,~~ \bar{A}_m=A_m+O(h), ~~ m \geq 1, \label{5.22a}\\
&{\rm{deg}}~\bar{B}_m=0,~~ \bar{B}_m=B_m+O(h), ~~ m\geq 2.
\end{align}
\end{subequations}
\end{proposition}

Thus, from the zero curvature representation \eqref{4.9} and \eqref{4.16} and
noticing the relation \eqref{A1-A2}, we have
\begin{align*}
\tilde{v}_{\bar{t}_m}=\bar{K}_m(\tilde{v})
&=(\bar{A}_{m,\bar{x}}-[\bar{A}_1, \bar{A}_m])\Delta^{-1},\\
&= \frac{1}{2}(A_{m,y}-[A_2,A_m])\partial_x^{-1} +O(h)\\
&= K_m+O(h),
\end{align*}
i.e.
\begin{align}\label{deg:K}
 {\rm{deg}}~\bar{K}_m=0,~~ \bar{K}_m=K_m+O(h), ~~ m=1,2, \cdots.
\end{align}
Similarly, from \eqref{4.16} we have
\begin{align}\label{deg:S}
 {\rm{deg}}~\bar{\sigma}_m=0, ~~  \bar{\sigma}_m=\sigma_m+O(h), ~~m=2,3,\cdots.
\end{align}

\subsection{Integrable structures and the symmetry constraint}\label{sec-5-3}

\subsubsection{Lax triads}\label{sec-5-2-1}

For the Lax triad \eqref{dmkp-iso-lax}, in addition to \eqref{5.12a} and \eqref{5.12b}, from \eqref{5.22a} we also have
\[\bar \phi_{\bar t_m}-\bar A_m \bar \phi=(\phi_{t_m}-A_m \phi)+O(h),\]
and the compatibility condition \eqref{dmkp-iso-com} yields
\begin{align*}
&\bar{L}_{\bar{x}}-[\bar{A}_1,\bar{L}]= \frac{h}{2}(L_y -[A_2, L])+ O(h^2),\\
&\bar{L}_{\bar{t}_m}-[\bar{A}_m,\bar{L}] =L_{t_m}-[A_m, L] +O(h), \\
&\bar{A}_{1,\bar{t}_m}-\bar{A}_{m,\bar{x}}+[\bar{A}_1,\bar{A}_m]
= \frac{h}{2}(A_{2,t_m}-A_{m,y}+[A_2,A_m])+O(h^2).
\end{align*}
For the nonisospectral case, we have
\begin{align*}
& \bar L \bar \phi - \lambda \bar \phi = (L\phi-\lambda \phi )+ O(h),\\
& \bar \phi_{\bar x}- \bar A_1 \bar \phi = \frac{h}{2}(\phi_y -A_2\phi)+ O(h^2),\\
& \bar \phi_{\bar t_m}-\bar B_m \bar \phi=(\phi_{t_m}-B_m \phi)+O(h),
\end{align*}
and
\begin{align*}
&\bar{L}_{\bar{x}}-[\bar{A}_1,\bar{L}]= \frac{h}{2}(L_y -[A_2, L])+ O(h^2),\\
&\bar{L}_{\bar{t}_m}-[\bar{B}_m,\bar{L}] =L_{t_m}-[B_m, L] +L^{m-1}+O(h), \\
&\bar{A}_{1,\bar{t}_m}-\bar{B}_{m,\bar{x}}+[\bar{A}_1,\bar{B}_m]
= \frac{h}{2}(A_{2,t_m}-B_{m,y}+[A_2, B_m])+O(h^2).
\end{align*}

\subsubsection{Symmetries and algebra}\label{sec-5-2-2}

The algebraic structures in \eqref{flow-alg-e} and \eqref{K-S-alg} are apparently different.
However, they agree with each other in the continuum limits.
To see that, taking \eqref{K-S-alg-2} as an example, we just need to calculate degrees of its two sides.
Using Proposition \ref{prop:deg-2}, \eqref{deg:K} and \eqref{deg:S}, we immediately arrive at
\begin{align*}
 \llbracket \bar{K}_m,\bar{\sigma}_n \rrbracket_{\tilde v}-(m (h \bar{K}_{m+n-1}+\bar{K}_{m+n-2}))
=  \llbracket  {K}_m, {\sigma}_n \rrbracket_{v}-m  {K}_{m+n-2}+O(h),
\end{align*}
which means  \eqref{flow-alg-e ksigma} is recovered from \eqref{K-S-alg-2} in the continuum limit.
In the following we skip presenting details and just list the results.

\begin{proposition}\label{prop:alg}
Except those relations involved with $\sigma_1$,
in the continuum limits the algebraic structure \eqref{flow-alg-e} of flows, symmetries
\eqref{corolary2  symmetries} and their structure \eqref{flow alge-1}
can be obtained from those of the D$\Delta$mKP case, i.e. \eqref{K-S-alg}, \eqref{4.26} and \eqref{4.27}.
\end{proposition}

\subsubsection{Squared eigenfunction symmetry constraint}\label{sec-5-2-3}

We can assume in our scheme of continuum limits, there are
\[a_n=q(x,y,t_m), ~~ b_n=r(x,y,t_m)\]
and $\deg a_n=\deg b_n=0$.
It then follows from \eqref{dmkp-L-com} that
\begin{equation*}
\bar{L}=h^{-1}(1+h a_n b_n)\Delta+ a_n \Delta^{-1} b_n \Delta
=\partial_x + q\partial_x^{-1}r\partial_x +O(h).
\end{equation*}
Next, in the spectral problem \eqref{delta-phi}, we assume
$\bar\varphi_{1,n}=\varphi_1(x,y,t_m), \,\bar\varphi_{2,n}=\varphi_2(x,y,t_m)$. Then
the spectral problem \eqref{KN-sp1} can be recovered from \eqref{delta-phi} in the continuum limit.
In addition, for the recursion operator $\bar T$ one can check that
\[\bar T= T +O(h).\]
Thus, the spectral problem and CLL hierarchy obtained from the mKP system
with the squared symmetry constraint are also obtained from the continuum limits.

\section{Conclusions}\label{sec-6}

In this paper we gave a new formulation of the squared eigenfunction symmetry constraint
to convert the D$\Delta$mKP system to the discrete CLL spectral problem and
the differential-difference CLL hierarchy.
This enabled us to introduce a unified continuum limit scheme and build correspondence
between the continuous and the differential-difference mKP hierarchies and their integrable structures.
We also introduced nonisospectral flows and obtained
two sets of  symmetries of the D$\Delta$mKP hierarchy and their Lie algebraic structure,
which recover all their continuous counterparts in the continuum limits.
With such correspondence, we have provided a relatively complete profile of the D$\Delta$mKP.

Note that the two different formulations of independent variables in the squared eigenfunction symmetry constraint
convert the D$\Delta$mKP system to the relativistic Toda ( cf.\cite{CZZ-2021}) and
the differential-difference CLL, respectively.
It would be worthy to understand the link between these two
(1+1)-dimensional differential-difference integrable systems.
In addition, the discrete CLL spectral problem \eqref{sd-CLL-sp}
is a Darboux transformation of the continuous CLL spectral problem (see Appendix \ref{App-2}).
In this context, one can investigate the connections between the
discrete CLL spectral problem and the fully discrete (potential) mKP equation,
which may lead to new algebraic geometry solution to the later, cf.\cite{XCZ-JPA-2022}.

\vskip 20pt
\subsection*{Acknowledgements}
This project is supported by the NSF of China (Nos. 12271334, 12126352, 12126343),
the Science and technology innovation plan of Shanghai (No. 20590742900)
and the Fundamental Research Funds for the Inner Mongolia Normal University (No. 2022JBBJ008).

\appendix

\section{Calculation of formula \eqref{vs-rhs}}\label{App-1}

By substitution of \eqref{4.37}  the right-hand side of \eqref{dmkp-vs}
and take $s=i$ yields
\begin{align*}
&\eqref{dmkp-vs}|_{\mathrm{r.h.s.}} \\
=& (-1)^sh^{-s}(ha_n b_n+1)(a_{n+1}-a_n)E^{-s-1}\Delta^sb_n
+(-1)^sh^{-s}a_nE^{-s-1}\Delta^s(ha_nb_n+1)(b_n-b_{n-1})\\
&~~ -(ha_n b_n+1)\Delta(-1)^sh^{-s}a_nE^{-s-1}\Delta^sb_n\\
&~~ +\sum_{i=1}^{s}(-1)^{s+1}\mathrm{C}_s^{i-1}h^{-s}
(a_nE^{-i-1}\Delta^i b_n) E^{-s-1}\Delta^{s-i+1}(ha_n b_n+1)\\
=&(-1)^s h^{-s}(A+ B+C),
\end{align*}
where $\eqref{dmkp-vs}|_{\mathrm{r.h.s.}}$ stands for the right-hand side of \eqref{dmkp-vs}, and
\begin{align*}
A= &  (\Delta a_n)E^{-s-1}\Delta^sb_n-\Delta  a_nE^{-s-1}\Delta^sb_n+a_nE^{-s-2}\Delta^{s+1}b_n,\\
B=&ha_n b_n(\Delta a_{n})E^{-s-1}\Delta^sb_n -ha_n b_n\Delta a_nE^{-s-1}\Delta^sb_n
+ha_nE^{-s-1}\Delta^sa_nb_n(\Delta b_{n-1}),\\
C=&-h\sum_{i=1}^{s} \mathrm{C}_s^{i-1} (a_nE^{-i-1}
\Delta^ib_n)E^{-s-1}\Delta^{s-i+1}a_n b_n.
\end{align*}

For the terms in $A$, using the formula
\begin{equation}\label{b.1}
\Delta f_n g_n=(Ef_n)\Delta g_n + (\Delta f_n) g_n,
\end{equation}
the first two terms together yield
$$-(Ea_n)E^{-s-1}\Delta^{s+1}b_n=-E a_nE^{-s-2}\Delta^{s+1}b_n.$$
Thus we have
\begin{equation*}
A=-\Delta a_nE^{-s-2}\Delta^{s+1}b_n.
\end{equation*}

For the terms in $B$, again, by using formula \eqref{b.1},  the first two terms yield
\begin{equation*}
-ha_n b_n(E a_{n})E^{-s-1}\Delta^sb_n=-ha_n b_n E a_{n} E^{-s-2}\Delta^sb_n.
\end{equation*}
For the third term in $B$, using \eqref{2.7} we have
\begin{align*}
& ha_nE^{-s-1}\Delta^sa_nb_n(\Delta b_{n-1})
=ha_n\Delta^s(\Delta b_{n-s-2})(E^{-s-1}a_nb_n)  \\
=&h\sum_{i=0}^{s} \mathrm{C}_s^i a_n (\Delta^{i+1} b_{n-i-2})(\Delta^{s-i}E^{-s-1}a_nb_n)\\
=&h\sum_{i=0}^{s} \mathrm{C}_s^ia_n(\Delta^{i+1} E^{-i-2}b_n)(\Delta^{s-i}E^{-s-1}a_nb_n).
\end{align*}
Replacing $i$ with $i-1$ in the above equation we have
\begin{align*}
& ha_nE^{-s-1}\Delta^sa_nb_n(\Delta b_{n-1}) \\
=&h\sum_{i=1}^{s+1} \mathrm{C}_s^{i-1}a_n(\Delta^i E^{-i-1}b_n)(\Delta^{s-i+1}E^{-s-1}a_nb_n)\\
=&h(a_n\Delta^{s+1}E^{-s-2}b_n)(E^{-s-1}a_nb_n) -C.
\end{align*}
Thus,
\begin{align*}
&\eqref{dmkp-vs}|_{\mathrm{r.h.s.}} \\
=&(-1)^s h^{-s}(A+ B+C)\\
=& (-1)^{s+1}h^{-s}[ (ha_n b_n+1)E(a_nE^{-s-2}\Delta^{s+1}b_n)
- (E^{-s-1}(ha_n b_n+1))(a_nE^{-s-2}\Delta^{s+1}b_n)],
\end{align*}
i.e. the equation \eqref{4.37}.

\section{On the discrete CLL spectral problem and equation}\label{App-2}

The discrete CLL spectral problem \eqref{sd-CLL-sp}, i.e.
\begin{equation}\label{sd-CLL-sp-app}
\left(\begin{array}{c} \psi_{1,n+1}\\  \psi_{2,n+1} \end{array}\right)=
\left(\begin{array}{cc}
-h\eta^2+1+h a_nb_{n+1}&   h  a_n\\
h \eta^2 b_{n+1}&1
\end{array}\right)
\left(\begin{array}{c}  \psi_{1,n}\\  \psi_{2,n} \end{array}\right),
\end{equation}
was found as early as in 1983 in \cite{DJM-1983},
where we have replaced $b_n$ with $b_{n+1}$.
It is compatible with the continuous spectral problem
\begin{equation}\label{c-CLL-sp-app}
\left(\begin{array}{c} \psi_{1,n}\\ \psi_{2,n} \end{array}\right)_{\bar x}=
\left(\begin{array}{cc} -\eta^2+a_nb_n & \eta a_n\\
 \eta  b_n& 0 \end{array}\right)
\left(\begin{array}{c} \psi_{1,n}\\ \psi_{2,n} \end{array}\right).
\end{equation}
The later was also given in \cite{DJM-1983} to lead to the unreduced CLL equations.
In addition, \eqref{c-CLL-sp-app} is also gauge equivalent to the usual CLL spectral problem (cf.\cite{WS-JPSJ-1983})
\begin{equation}\label{c-CLL-sp-app-2}
\left(\begin{array}{c} \phi_{1,n}\\ \phi_{2,n} \end{array}\right)_{\bar x}=
\left(\begin{array}{cc}
\frac{1}{2}(-\eta^2+a_nb_n) & \eta a_n\\
 \eta  b_n& -\frac{1}{2}(-\eta^2+a_nb_n) \end{array}\right)
\left(\begin{array}{c} \phi_{1,n}\\ \phi_{2,n} \end{array}\right).
\end{equation}
by taking $\phi_j=\psi_j e^{\frac{1}{2}(\eta^2\bar x^2 -\partial_{\bar x}^{-1}a_nb_n)}$.
The compatibility between \eqref{sd-CLL-sp-app} and \eqref{c-CLL-sp-app}
gives rise to a differential-difference CLL equation
\begin{equation*}
	a_{n,\bar{x}}=h^{-1}(1+h a_nb_n)(a_{n+1}-a_n),~~b_{n,\bar{x}}=h^{-1}(1+h a_nb_n)(b_n-b_{n-1}),
\end{equation*}
i.e. \eqref{4.35}.
Meanwhile, the compatibility allows us to view \eqref{sd-CLL-sp-app} as a
Darboux transformation of the CLL spectral problem \eqref{c-CLL-sp-app}.

It is well known that transformation is one of means of discretization \cite{LB-PNAS-1980}.
Based on the above relations, if we denote the Darboux matrix in \eqref{sd-CLL-sp-app} as
\begin{equation}
M(\gamma, h, a_{n,m}, b_{n+1,m})=
\left(\begin{array}{cc}
-h\gamma +1+h a_{n,m}b_{n+1,m}& h  a_{n,m}\\
h \gamma  b_{n+1,m}& 1
\end{array}\right),
\end{equation}
then the compatibility of
\[\Psi_{n+1,m}=M(\gamma, p, a_{n,m}, b_{n+1,m}) \Psi_{n,m},~~
\Psi_{n,m+1}=M(\gamma, q, a_{n,m}, b_{n,m+1}) \Psi_{n,m},
\]
where $\Psi_{n,m}=(\psi_1(n,m),\psi_2(n,m))^T$, yields the discrete CLL equation \cite{DJM-1983}
\begin{align*}
& p a_{n,m+1}-q a_{n+1,m}-(p-q) a_{n,m}=pq(a_{n+1,m}-a_{n,m+1})a_{n,m}b_{n+1,m+1},\\
& q b_{n,m+1}-q b_{n+1,m}+(p-q) b_{n+1,m+1}=pq(b_{n+1,m}-b_{n,m+1})a_{n,m}b_{n+1,m+1}.
\end{align*}

It is also notable  that the discrete CLL spectral problem \eqref{delta-psi} is a direct discretization
of the continuous one \eqref{c-CLL-sp-app}
by just replacing $(\psi_{i,n})_x $ with $(\psi_{i,n+1}-\psi_{i,n})/h$ for $i=1,2$.
This is interesting because the Merola-Ragnisco-Tu spectral problem,
which can be obtained  from the squared-eigenfunction-symmetry-constrained D$\Delta$KP system,
can also be viewed as a direct discretization as well as a Darboux transformation
of the AKNS spectral problem, see Eq.(1.5) and Appendix B in \cite{Chen-JNMP-2017}.

\end{document}